\documentclass[accepted=2019-01-30,onecolumn]{quantumarticle}

\usepackage{enumitem}
\usepackage{color}
\usepackage{graphicx, amssymb, amsmath, fullpage, dsfont, amsthm}
\usepackage{authblk}
\usepackage[colorlinks=true,urlcolor=webblue,linkcolor=webgreen,filecolor=webblue,citecolor=webgreen,pdfpagemode=UseOutlines,pdfstartview=FitH,pdfpagelayout=OneColumn,bookmarks=true,backref=page]{hyperref}

\newif\ifbackrefshowonlyfirst
\backrefshowonlyfirstfalse
\makeatletter
\let\BR@direct@old@hyper@natlinkstart\hyper@natlinkstart
\renewcommand*{\hyper@natlinkstart}{\phantomsection\BR@direct@old@hyper@natlinkstart}
\let\BR@direct@oldBR@citex\BR@citex
\renewcommand*{\BR@citex}{\phantomsection\BR@direct@oldBR@citex}%

\long\def\hyper@page@BR@direct@ref#1#2#3{\hyperlink{#3}{#1}}

\ifx\backrefxxx\hyper@page@backref
    \let\backrefxxx\hyper@page@BR@direct@ref
    \ifbackrefshowonlyfirst
    \fi
\else
    \ifbackrefshowonlyfirst
    \fi
\fi

\RequirePackage{etoolbox}
\patchcmd{\Hy@backout}{Doc-Start}{\@currentHref}{}{\errmessage{I can't seem to patch backref}}
\makeatother

\newenvironment{algobox}[1]
{\begin{center}\begin{minipage}{#1\textwidth}\hrulefill\\}
{\vspace{-8pt}\hrulefill\end{minipage}\end{center}}

\newcommand{\alg}[1]{\textbf{#1}}

\newcommand{\QFTPermModtxt}{\alg{QFTPermMod}}      \newcommand{\QFTPermMod}{\hyperref[alg:QFTPermMod]{\QFTPermModtxt}}
\newcommand{\DualSchurtxt}{\alg{DualSchur}} \newcommand{\DualSchur}{\hyperref[alg:DualSchur]{\DualSchurtxt}}
\newcommand{\GPEtxt}{\alg{GPE}} \newcommand{\GPE}{\hyperref[alg:GPE]{\GPEtxt}}

\usepackage{tikz}
\usetikzlibrary{patterns,snakes}

\usetikzlibrary{backgrounds,fit,decorations.pathreplacing,calc}
\usepackage{subfigure,caption}
\usepackage{xcolor}
\usepackage[enableskew]{youngtab}
\usepackage{genyoungtabtikz}
\usepackage{fullpage}
\usepackage[numbers]{natbib}

\newcommand{\captionfonts}{\small}

\makeatletter  
\long\def\@makecaption#1#2{%
  \vskip\abovecaptionskip
  \sbox\@tempboxa{{\captionfonts #1: #2}}%
  \ifdim \wd\@tempboxa >\hsize
    {\captionfonts #1: #2\par}
  \else
    \hbox to\hsize{\hfil\box\@tempboxa\hfil}%
  \fi
  \vskip\belowcaptionskip}
\makeatother   

\colorlet{darkblue}{blue!70!black}
\colorlet{darkred}{red!70!black}
\colorlet{lightgrey}{black!25}
\colorlet{grey}{black!50}

\hypersetup{%
  pdftitle=An efficient high dimensional quantum Schur transform,
  pdfauthor=Hari Krovi}

\definecolor{webgreen}{rgb}{0,.5,0}
\definecolor{webblue}{rgb}{0,0,.5}

\numberwithin{equation}{section}

\newtheorem{theorem}{Theorem}
\newtheorem{lemma}[theorem]{Lemma}

\newcommand{\remove}[1]{}

\newcommand{\C}{\mathbb{C}}

\newcommand{\U}{\textsf{U}}

\newcommand{\ket}[1]{|#1\rangle}
\newcommand{\bra}[1]{\langle #1|}

\newcommand{\SU}{\textsf{SU}}

\newcommand{\Ind}[2]{\operatorname{\uparrow}_{#1}^{#2}}

\newcommand{\triv}{\mathbf{1}}
\newcommand{\be}{\begin{equation}}
\newcommand{\ee}{\end{equation}}

\begin{document}

\title{An efficient high dimensional quantum Schur transform}
\author{Hari Krovi}
\orcid{0000-0001-9675-9959}
\email{hkrovi@gmail.com, hari.krovi@raytheon.com} 
\affil[]{Quantum Engineering and Computing\\Physical Sciences and Systems\\Raytheon BBN Technologies, Cambridge, MA}

\maketitle

\begin{abstract}
The Schur transform is a unitary operator that block diagonalizes the action of the symmetric and unitary groups on an $n$ fold tensor product $V^{\otimes n}$ of a vector space $V$ of dimension $d$. Bacon, Chuang and Harrow \cite{BCH07} gave a quantum algorithm for this transform that is polynomial in $n$, $d$ and $\log\epsilon^{-1}$, where $\epsilon$ is the precision. In a footnote in Harrow's thesis \cite{H05}, a brief description of how to make the algorithm of \cite{BCH07} polynomial in $\log d$ is given using the unitary group representation theory (however, this has not been explained in detail anywhere\footnote{Specifically, Harrow sketches a procedure to obtain the same labeling of irreps as the ones used for Gelfand-Tsetlin basis. Why this procedure leads to the Gelfand-Tsetlin basis itself (not just the same labeling scheme) could be explained in more detail.}). In this article, we present a quantum algorithm for the Schur transform that is polynomial in $n$, $\log d$ and $\log\epsilon^{-1}$ using a different approach. Specifically, we build this transform using the representation theory of the symmetric group and in this sense our technique can be considered a ``dual" algorithm to \cite{BCH07}. A novel feature of our algorithm is that we construct the quantum Fourier transform over the so called \emph{permutation modules}, which could have other applications.
\end{abstract}


\section{Introduction}
Schur-Weyl duality is a remarkable correspondence between the irreducible representations of the symmetric group and those of the unitary group acting on an $n$ fold tensor product of a vector space $V$. This correspondence allows one to construct all of the so-called polynomial representations of the unitary, general linear and special linear groups. Polynomial representations of matrix groups such as unitary groups are representations whose matrix entries can be written as polynomials in the entries of the group element i.e., $\rho(U)$ is a polynomial representation if the entries of $\rho(U)$ are polynomial in the entries of $U$. Schur-Weyl duality has been generalized to many other groups and algebras including quantum groups \cite{dipper2008brauer, Halverson}.

Schur-Weyl duality has numerous applications in quantum information theory. It has been used to prove that the tensor product of many copies of a density operator is close to a projector. In fact, the projector is the one corresponding to a partition in Schur-Weyl duality that is closest to the spectrum of $\rho$ \cite{ARS88,KW01,HM02a, CM06}. It has also been used to prove de Finetti theorems \cite{KM09}, which have many applications in security proofs of quantum key distribution systems. The Schur transform was first constructed for qubits in the work of Bacon, Chuang and Harrow in \cite{BCH06}. This has been extended to qudits by the same authors in \cite{BCH07}. A quantum circuit for the Schur transform also has numerous applications in quantum information theory. It has been applied to universal distortion-free entanglement concentration \cite{HM07}, universal compression \cite{HM02a,HM02b}, encoding and decoding into decoherence-free subspaces \cite{ZR97, KLV00, KBLW01, Bac01}. These applications and others are discussed in more detail in Harrow's thesis \cite{H05}. Recently, the Schur transform has been used as a primitive in an efficient algorithm for spectrum testing of a density operator \cite{OW15a} and in algorithms for sample optimal state tomography \cite{OW15b,HHJW15} improving on previous algorithms \cite{VLPT99,GM02,Hay02b,Key04}. It is also used in a scheme for optimal compression \cite{YCH16}, quantum superreplication \cite{CY16,CYH15}. Recently, in \cite{HS18}, the authors present a classical algorithm to compute the transition amplitudes of the Schur sampling circuits.

There are two main ways in which Schur sampling is used: weak and strong sampling. Recall that with any kind of Fourier sampling, we block diagonalize the group action to obtain the Fourier basis from the computational basis. Strong (resp. weak) sampling refers to measuring all the registers (resp. only the label of the irreducible representation) to obtain information from the Fourier basis. Similarly, a circuit for Schur sampling block diagonalizes the unitary group (or equivalently the symmetric group) representation on the $n$ fold tensor power of a $d$ dimensional space. In this block diagonalization, each block essentially contains three pieces of information: (a) the irrep label, which is common to both the unitary and symmetric groups, (b) the unitary group irrep register i.e., the register holding the state that transforms under the unitary group and (c) the symmetric group irrep register holding the state that transforms under the symmetric group. Strong Schur sampling refers to measuring all three registers and weak Schur sampling refers to measuring only the irrep label. One could always measure all three registers in some basis to obtain more information about the state being transformed. However, in most of the above applications, one only needs to measure the irrep label register. In other words, one only needs weak Schur sampling and the associated probability distribution over the irrep labels to get enough information for all of the above mentioned applications. Weak Schur sampling can be performed with a quantum circuit that is polynomial in $\log d, n$ and $\log(1/\epsilon)$ using the so-called Generalized Phase Estimation (GPE) procedure. The construction of such a circuit for the weak Schur transform is described in \cite{H05}.

In this paper, we present a circuit that can be used for strong Schur sampling using the representation theory of the symmetric group that runs in time polynomial in $\log d, n$ and $\log(1/\epsilon)$. In the previous approach by Bacon, Chuang and Harrow (BCH) \cite{BCH07}, the Schur transform is constructed using the unitary group representation theory. By contrast, our algorithm uses the symmetric group and, in fact, uses transforms such as Beals' algorithm for the Fourier transform over the symmetric group \cite{Beals97} and the GPE algorithm mentioned above. To give an intuition for how the dual Schur transform works, we should look at how the original algorithm by BCH for the Schur transform works. BCH developed a quantum circuit for the Clebsch-Gordan (CG) decomposition problem for the unitary group, which entails block diagonalizing a tensor product of two irreps of the unitary group. Then they use this circuit to construct the Schur transform by applying it iteratively. 

This gives us a clue as to how to construct the dual transform if we look at the so called Littlewood-Richardson (LR) coefficients. The LR coefficients give the number of times an irrep of the unitary group appears in the CG decomposition of two unitary group irreps. The same coefficients describe the decomposition of induced representations of the symmetric group from certain Young subgroups. This suggests that we need to investigate permutation modules of the symmetric group since they are exactly such induced representations. We can now see that block diagonalizing these induced representations gives us the dual Schur transform. We will not directly work with LR coefficients or the related Kostka numbers, but embed them in a larger space. This dual version is polynomial in $\log d, n$ and $\log(1/\epsilon)$, where as the one in \cite{BCH07} is polynomial in $d, n$ and $\log(1/\epsilon)$. However, as mentioned above, in Harrow's thesis \cite{H05}, a short description is given of how to make the algorithm from \cite{BCH07} polynomial in $\log d, n$ and $\log(1/\epsilon)$. A classical algorithm to compute the CG coefficients for \SU(d) is given in \cite{CG_coeff}.

This paper is organized as follows. In Section \ref{Sec:Rep_theory_background}, we explain the necessary background in representation theory, namely the notions of induced representations, symmetric and unitary group representations and the structure of permutation modules of the symmetric group. In Section \ref{Sec:QFTs}, we describe the construction of quantum Fourier transforms over induced representations. This construction was also used in \cite{KR15} for Fourier transforms over quantum doubles over finite groups. Then, we apply this QFT for permutation modules which are essentially induced representations. In Section \ref{Sec:Dual_Schur_transform}, we construct the dual algorithm for the Schur transform and show that the construction is O(poly$(n,\log d, \log 1/\epsilon))$ elementary operations. Finally, in Section \ref{Sec:Conclusions}, we present some conclusions and other potential applications of our subroutines.

\section{Background in representation theory}\label{Sec:Rep_theory_background}
\subsection{Basics of induced representations}
In this section, we briefly describe the representation theoretic concepts such as irreducible representations (irreps, for short), regular representations and induced representations. Induced representations are important in this article since the dual Schur transform is essentially a block diagonalization of induced representations. These concepts for finite groups are described in several texts such as \cite{Serre77, FH91}. Here we follow the development in \cite{FH91}.

A representation of a finite group on a finite dimensional complex Hilbert space $V$ is a homomorphism from the group $G$ to the unitary group on the vector space $\U(V)$ i.e., a representation is $\rho: G \rightarrow \U(V)$. For every finite or compact group, any representation on $V$ can be made unitary i.e., $\rho(g)$ is a unitary matrix for all $g\in G$. Very often, the space which the representation $\rho$ maps to, is identified with the representation. Two representations $\rho$ and $\rho^\prime$ of a group $G$ acting on the same vector space $V$ are considered equivalent if there exists a unitary $U$ such that $U\rho(g)U^\dag = \rho^\prime(g)$ for all $g\in G$. A subspace $W$ of $V$ is called a sub-representation if $\rho(g)$ preserves $W$ for all $g\in G$. In this case, the orthogonal complement of $W$ in $V$ is also a sub-representation and $V$ can be viewed as a direct sum of these two sub-representations. A representation is called an irreducible representation (irrep) if it does not contain any non-trivial sub-representations. Any representation $V$ can be broken up into a direct sum of sub-representations $W$ and its complement as above. Continuing this process further and breaking up $W$ and its complement into sub-representations, one can arrive at a decomposition of $V$ into irreps: $V\cong V_1\oplus \dots \oplus V_n$, where some of the irreps in the decomposition may be equivalent. A special kind of irrep is the trivial irrep which acts on a one-dimensional vector space and takes all the group elements to the identity. Any finite group $G$ has a finite number of irreducible representations whose number is equal to the number of conjugacy classes of $G$. 

A regular representation of a finite group $G$ acts on the vector space $\C[G]$, where $g$ acts on any basis vector $h$ by left multiplication $g:h\rightarrow gh$. The regular representation turns out to have the following interesting direct sum decomposition into irreps.
\begin{equation}
\label{Eq:regular-rep}
\C[G] \cong \bigoplus_{i} W_i^{\oplus d_i}\,,
\end{equation}
where $W_i$ is an irrep of $G$, $i$ runs over all the different irreps of $G$ and $d_i$ is the vector space dimension of the irrep $W_i$. The quantum Fourier transform (QFT) over a group $G$ usually refers to the transform that performs the above block diagonalization. This can be defined as the following basis transformation.
\[
\ket{g}\rightarrow \ket{\rho,i,j}\,,
\]
where $\rho$ is the label of the irrep, $i$ is the multiplicity space index and $j$ is the irrep space index. From the above decomposition, we can see that the dimension of both the multiplicity space and irrep space are the same and so $i$ and $j$ run over the same index set labeling the basis vectors.

A type of representation that is of particular importance here is the induced representation defined as follows. Given a subgroup $H$ of $G$ and a representation $(\rho,W)$ of $H$, we can construct a representation $V$ of $G$ as follows. As a vector space, it is the tensor product $\C[G/H]\otimes W$, where $\C[G/H]$ is the space of cosets of $H$ in $G$. The action of $G$ on this basis can be described using a transversal $G/H = \{t_1,t_2,\dots ,t_m\}$ for $H$ in $G$, where $m$ is the number of cosets i.e., $m=|G|/|H|$. This means that these elements form an orthonormal basis of the vector space $\C[G/H]$. Let $\{\ket{w_1},\dots ,\ket{w_d}\}$ be a basis of $W$. We denote the induced representation by $\Ind{H}{G}\rho$ or $\Ind{}{G}\rho$ (when $H$ can be inferred). In this basis, $(\Ind{H}{G} \rho)(g)$ is the following action on basis vectors (which can be linearly extended to other vectors).
\begin{equation}
\label{eq:induced-representation}
(\Ind{H}{G} \rho)(g):\ket{t} \otimes \ket{w_i} \mapsto \ket{t'} \otimes \rho(h)\ket{w_i}
\end{equation}
where
$t' \in G/H$ and $h \in H$ are the unique elements for which $gt = t'h$.

\subsection{Irreducible representations of symmetric and unitary groups}
The symmetric group on $n$ letters is denoted $S_n$ and consists of all possible permutations of the $n$ letters. There are $n!$ permutations and every element can be written as a product of transpositions, where a transposition is a swap of two letters. Here we denote a transposition between $a$ and $b$ as $(a,b)$. The representation theory of the symmetric group is discussed in several books (see for example, \cite{FH91, JK81, Sag13}). The irreducible representations of $S_n$ are labeled by Young diagrams, which are diagrams that consist of rows of boxes. A Young diagram corresponds to a partition of $n$, which is defined as a tuple $\lambda=(\lambda_1,\lambda_2,\dots \lambda_k)$, where $\lambda_j\geq 0$, $\sum_j \lambda_j=n$ and $\lambda_k\geq \lambda_l$ for $k<l$. Given a partition, a Young diagram has $k$ rows and $i_j$ boxes in row $j$. In this case, we say that there are $k$ parts in the partition $\lambda$. For example, a Young diagram corresponding to the partition  $(2,1)$ (not to be confused with a transposition) is given by ${\Yboxdim{13pt}\Yvcentermath1\yng(2,1)}$. A Young tableau is a Young diagram with numbers in the boxes. If the numbers are from $1$ to $n$, increasing from left to right and increasing from top to bottom, then the Young tableau is called a standard Young tableau (SYT). For example, for the partition $(2,1)$, an SYT could be ${\Yboxdim{13pt}\Yvcentermath1\young(13,2)}$ or ${\Yboxdim{13pt}\Yvcentermath1\young(12,3)}$. In fact, these are the only possible choices. This reflects the fact that the irrep labeled by this Young diagram is two dimensional. In general, the dimension of the irrep of $S_n$ corresponding to a partition $\lambda$ is the number of possible standard Young tableau. It is also given by the famous hook length formula
\begin{equation}
d_\lambda = \frac{n!}{\Pi_{i,j} \, h_\lambda (i,j)}\,,
\end{equation}
where the product in the denominator is over all boxes $(i,j)$ and $h_\lambda(i,j)$ is the hook length of a box, which is defined as the sum of the number of boxes to the right (including the box $(i,j)$) and the number of boxes directly below the given box $(i,j)$. We will get back to other aspects of the symmetric group when we discuss subgroup adapted bases and permutation modules in the next section.

The unitary group $\U(d)$ is the group of $d\times d$ unitary matrices that is an infinite, though compact, group. The irreducible representations of this group can be labeled in several ways. Here, we describe the labeling using both Dynkin labels and Young diagrams. A Dynkin label is the set of coefficients in the so called basis of fundamental weights. Every irreducible representation has a basis whose vectors are called weight vectors and have weights associated with them. In this basis, the highest (and lowest) weight vectors are unique and, in fact, every irrep of the unitary group can be associated to a unique highest weight (and there is a one-to-one correspondence between them). It turns out that the weight vectors lie in a space (called the root space) spanned by the fundamental weights. This means that every highest weight (of any irrep) can be written as a linear combination of the fundamental weights. The weights and weight vectors of a given irrep form a special basis of the irrep of the unitary group called the Gelfand-Tsetlin basis. We will discuss this basis in more detail below. It turns out that one can use Young diagrams to label irreps of the unitary group as well. One can convert the Dynkin label representation to a Young diagram representation in the following way: if the Dynkin labels are $(l_1, \dots l_r)$, then the corresponding partition $\lambda$ has components $\lambda_i=l_i+\dots +l_r$. Conversely, a Young diagram $\lambda$ that represents an irrep can be converted to Dynkin labels by setting $l_i=\lambda _i-\lambda_{i-1}$.

\subsection{Subgroup adapted bases}\label{subsec:subgroup_adapted_bases}
A subgroup adapted basis is a canonical basis for an irrep of a group $G$ that is obtained from a tower of subgroups $G_0=1\subset G_1\subset\dots G_n=G$ from the identity to $G$. To see how one obtains a canonical basis from a given subgroup tower, first consider an irrep $\rho$ of $G=G_n$. Suppose we restrict it to the subgroup $G_{n-1}$, then $\rho$ can be decomposed into irreps of $G_{n-1}$. Suppose that this restriction yields irreps $\sigma_i$ of $G_{n-1}$ each with multiplicity $m_i$, then choosing a basis for the multiplicity spaces and the irrep spaces would give us a basis for $\rho$. In choosing a basis for the $\sigma_i$, we can restrict to $G_{n-2}$ and so on down the subgroup tower. Finally, we would end up with the trivial subgroup and since it has only a one-dimensional irrep, this would fix the entire basis. In the special case that each of the restrictions from $G_i$ to $G_{i-1}$ are multiplicity-free i.e., $m_i$ are all zero or one, we get a canonical basis. In other words, there is no ambiguity in choosing the basis for the multiplicity space except for the choice of a phase for each multiplicity space.

In most applications, one makes a projective measurement in this basis, which makes this phase choice irrelevant. However, if one were to perform other quantum operations conditioned on basis vectors after the Schur transform, then the phase choice might be relevant. But in that case, one can incorporate that phase choice into the conditional operations to be performed. We will see next that there are special subgroup towers for both the symmetric and unitary groups that have multiplicity-free branching along the tower and hence lead to canonical bases.

For the symmetric group, the tower of subgroups $1=S_1\subset S_2\subset\dots S_n$, where $S_i$ permutes the first $i$ letters and fixes the remaining $n-i$ ones, gives a multiplicity-free branching rule from one subgroup to the next. This tower is fixed once we number the $n$ qudit registers in some fashion. The resulting canonical basis is called Young orthonormal basis (also Young-Yamanouchi basis). This basis can be associated to Young diagrams with numbers in the boxes with the rule that the numbers are strictly increasing as one goes from left to right along a row and top to bottom along a column. As mentioned above, such numbered Young diagrams are called \emph{standard Young tableaux} (SYT). An example is given below.
\begin{equation}
\lambda = \Yvcentermath1\young(134,256,7)\,.
\end{equation}
As is well-known, the symmetric group is generated by adjacent transpositions $(k,k+1)$. If $k$ and $k+1$ are in different rows and columns in $T$, the action of any such transposition on an SYT $T$ is given as
\begin{align}
&(k,k+1)\ket{T}=a^T_k\ket{T}+b^T_k\ket{(k,k+1)T}\,,\\
&(k,k+1)\ket{(k,k+1)T}=b^T_k\ket{T}-a^T_k\ket{(k,k+1)T}\,,
\end{align}
where $a^T_k$ is the inverse of the Manhattan distance in $T$ between the boxes labeled $k$ and $k+1$ and $b^T_k=\sqrt{1-(a^T_k)^2}$. The Manhattan distance between two boxes is the number of steps needed to go up plus number of steps to the right minus number of steps to the left minus number of steps down. It can be seen easily that this does not depend on the path taken. We use the notation $\ket{(k,k+1)T}$ to denote the SYT with $k$ and $k+1$ interchanged, which can be seen to be a SYT. If $k$ and $k+1$ are in the same row, they must be next to each other and the action is given as
\be
(k,k+1)\ket{T}=\ket{T}\,,
\ee
and if they are in the same column (and necessarily in adjacent rows) the action is 
\be
(k,k+1)\ket{T}=-\ket{T}\,.
\ee

For the unitary group, a subgroup tower that leads to a canonical basis is $1=U_1\subset U_2\subset \dots U_d$, where $U_i$ is the unitary group acting on the $i\times i$ minor of the full $d\times d$ matrix. This tower is determined once we fix a basis for each qudit register. This tower, like the one for the symmetric group, gives rise to multiplicity-free branching and hence to a canonical basis. This canonical basis is called the Gelfand-Tsetlin basis. Given any irrep $\lambda$ of $U_d$ in the form of a Young diagram, one can obtain the diagrams in the restriction to $U_{d-1}$ by removing a box from end of each column in all possible ways. If two columns have the same length in $\lambda$ and the choice is to remove a box from the left column, then one must also remove the box from the right column to ensure that a valid Young diagram is obtained. An example is given below.
\begin{equation}
\lambda = \Yvcentermath1\yng(3,3,1)\rightarrow \left\{\Yvcentermath1\yng(3,3,1), \Yvcentermath1\yng(3,3),\Yvcentermath1\yng(3,1,1), \Yvcentermath1\yng(3,1) \right\}\,.
\end{equation}
Suppose we pick one and proceed with a choice down the tower, we would have the following possibility.
\begin{equation}
\Yvcentermath1\yng(3,3,1)\rightarrow \Yvcentermath1\yng(3,3,1)\rightarrow\Yvcentermath1\yng(3,1)\rightarrow\Yvcentermath1\yng(2)\,.
\end{equation}
This sequence gives us a basis vector of the irrep $\lambda$. This can be encapsulated by putting numbers into the original irrep, which represent the stage before which the boxes are removed. For the above sequence the following numbering would hold.
\begin{equation}
\Yvcentermath1\young(112,233,3)\,.
\end{equation}
Notice that the rows are weakly increasing i.e., the numbers either increase or stay the same as we move right and the columns are strictly increasing. Such a Young diagram is called a semi-standard Young diagram (SSYT). SSYTs with numbers taken from the set $[d]$ label basis vectors in the irrep of $U_d$. We will see below that SSYTs also play a role in certain induced representations of the symmetric group called permutation modules. Efficient encodings of these bases are constructed in \cite{BCH07} i.e., using poly($\log d$, n, $\log 1/\epsilon$) bits. For a SSYT, the tuple that contains the number of boxes labeled by a given integer is called the \emph{content} of the SSYT. For instance, in the example above, the content is $(2,2,3)$ corresponding to $2$ boxes numbered one, $2$ boxes numbered two and $3$ boxes numbered three.

SSYTs have an interesting structure that is useful in our algorithms. If we consider all the boxes containing a specific number, we find that no two of them appear in the same column. If we isolate these boxes, such a \emph{skew} diagram is called a \emph{horizontal strip}. An SSYT can be thought of as being composed of horizontal strips. It turns out that this composition can be made more precise as we describe in the next subsection. An example of an SSYT and the associated horizontal strips are given below.

\begin{equation}
\Yvcentermath1\gyoung(112,233,3)\,,\quad \Yvcentermath1\gyoung(11)\,,\quad \Yvcentermath1\gyoung(::;2,;2)\,,\quad\Yvcentermath1\gyoung(:;3;3,;3)\,.
\end{equation}
In terms of horizontal strips, the above decomposition of an irrep $\lambda$ of $U(d)$ into irreps $\mu$ of $U(d-1)$ can be rephrased as the set of all $\mu$ that can be obtained from $\lambda$ by removing a horizontal strip in all possible ways.

\subsection{RSK algorithm and composition of Young tableaux}\label{subsec:RSK}
The discussion in this section is taken from the book by Fulton \cite{fulton1997young}. This procedure is used in the algorithm \DualSchur ~(step 4). For a more detailed explanation of how to obtain SYTs with permuted content, see \cite{fulton1997young} chapter 4. The RSK (Robinson-Schensted-Knuth) algorithm establishes a correspondence between pairs of words and pairs of tableaux. The main part of the RSK is a procedure called row insertion that lets one insert a letter into a tableau such that the resulting tableau is semi-standard with one more box. This correspondence has several applications, but the primary application here is to produce semi-standard Young tableaux where the content is permuted. In this subsection, we briefly describe this algorithm. The row insertion procedure takes as input a tableau $T$ and an integer $x$ and outputs a tableau with one more box than $T$. The procedure is as follows 
\begin{enumerate}
\item Find the number in the first row of $T$ that is greater than $x$. 
\item If there is none, then place $x$ in a box at the end of the first row.
\item If there are numbers greater than $x$ in the first row, let $y$ be the smallest among them. Place $x$ in $y$'s position ($x$ `bumps' $y$).
\item Repeat the previous steps with $y$ (in the place of $x$) and starting with the second row.
\end{enumerate}
The RSK algorithm uses the above row bumping procedure. It takes a pair of words, say $u=u_1u_2\dots u_r$ and $v=v_1v_2\dots v_r$ that has the following two properties as input. The first is that $u$ is weakly increasing and second if $u_{k-1}=u_k$, then $v_{k-1}\leq v_k$. Given such pairs as input the procedure produces a pair of tableaux ($P,Q$) iteratively as follows. Start with the base tableaux ${\scriptsize\Yvcentermath1\young(x)}$ and ${\scriptsize\Yvcentermath1\young(y)}$, where $x=v_1$ and $y=u_1$. Then from any pair ($P_{k-1},Q_{k-1}$), row insert $v_k$ into $P_{k-1}$ getting $P_k$. Then add a box to $Q_{k-1}$ in the position where the new box is in $P_k$ and put $u_k$ in this box.

This procedure allows us to define a product or composition of tableaux mentioned in the previous subsection. Suppose $S$ and $T$ are two tableaux, then $S\cdot T$ is defined as follows. If $T$ consists of only one box, then the product of $S$ and $T$ is the result of row insertion into $S$. If $T$ contains more than one box, then we row insert them one by one into $S$ starting from the bottom left box and moving left to right along each row and upwards along the rows. 

In this paper, we will need this procedure (in the algorithm \DualSchur ~step 4) to create an SSYT $V$ from another SSYT $U$ with the content permuted. In order to describe it for any permutation, we only need to show it for a single transposition. As noted above, an SSYT consists of horizontal strips each having the same number. Now suppose that the transposition is $(k,k+1)$ i.e., if the original SSYT $U$ contains $n_k$ boxes numbered $k$ and $n_{k+1}$ boxes numbered $k+1$, then the new SSYT $V$ should contain $n_{k+1}$ boxes numbered $k$ and $n_k$ boxes numbered $k+1$. This is done by using the product defined above. It turns out \cite{fulton1997young} that we can write $U=A\cdot B\cdot C$, where $A$ is a SSYT that contains only boxes numbered $1$ through $k$, $B$ is an SSYT that contains boxes numbered $k$ and $k+1$ and $C$ is an SSYT that contains the remaining numbers. Since $B$ contains only two labels, it must be of the form
\Yboxdim{1cm}
\begin{center}
\begin{tikzpicture}[scale=0.7]
\tgyoung(4.5cm,0.2cm,k_2\hdts kk_2\hdts kl_2\hdts l,l_2\hdts l)
\draw [
    thick,
    decoration={
        brace,
        mirror,
        raise=0.5cm
    },
    decorate
] (8.7,0.5) -- (12.5,0.5)
node [pos=0.5,anchor=north,yshift=-0.55cm] {s};
\draw [
    thick,
    decoration={
        brace,
        mirror,
        raise=0.5cm
    },
    decorate
] (12.6,0.5) -- (16.5,0.5)
node [pos=0.5,anchor=north,yshift=-0.55cm] {t};
\end{tikzpicture}
\end{center}
where $l=k+1$ for brevity in the figure and the overhang consists of $s$ boxes numbered $k$ and $t$ boxes numbered $l$. When we swap the number of $k$s and $l$s, we obtain a similar diagram where the overhang contains $t$ $k$s and $s$ $l$s. Denote this diagram by $B^\prime$. 
\Yboxdim{1cm}
\begin{center}
\begin{tikzpicture}[scale=0.7]
\tgyoung(4.5cm,0.2cm,k_2\hdts kk_2\hdts kl_2\hdts l,l_2\hdts l)
\draw [
    thick,
    decoration={
        brace,
        mirror,
        raise=0.5cm
    },
    decorate
] (8.7,0.5) -- (12.5,0.5)
node [pos=0.5,anchor=north,yshift=-0.55cm] {t};
\draw [
    thick,
    decoration={
        brace,
        mirror,
        raise=0.5cm
    },
    decorate
] (12.6,0.5) -- (16.5,0.5)
node [pos=0.5,anchor=north,yshift=-0.55cm] {s};
\end{tikzpicture}
\end{center}

Now the new SSYT $V$ is obtained by composing $A\cdot B^\prime\cdot C$ using the RSK algorithm.

\subsection{Gelfand-Tsetlin bases}\label{subsec:GT}
An alternate way of representing the basis vectors of the unitary group is the so called Gelfand-Tsetlin (GT) patterns. GT patterns are useful in certain applications, although one can easily convert an SSYT to a GT pattern and vice versa. The power of GT patterns comes from the fact that in the GT basis, one can write the matrix elements of the Lie algebra \SU(d) as derived by Gelfand and Tsetlin \cite{GT}. The formulae in this section can be found in the book by Vilenkin and Klimyk \cite{Klymik}.

A GT pattern $M$ is a triangle of numbers such as the one below.
\begin{equation}
M=\begin{pmatrix}
m_{1,d} & & m_{2,d} & & \dots & & m_{d,d}\\
 & m_{1,d-1}& &  \dots & & m_{d-1,d-1} &\\
&& \ddots &&\reflectbox{$\ddots$}\\
&&m_{1,2}&m_{2,2}\\
&&\qquad \qquad \quad m_{1,1}
\end{pmatrix}\,.
\end{equation}
These numbers satisfy the \emph{in betweenness} condition
\begin{equation}
m_{k,l}\geq m_{k,l-1}\geq m_{k+1,l}\,,\qquad 1\leq k<l\leq d\,.
\end{equation}
The numbers in the first row of the GT pattern correspond to the number of boxes in each row of the corresponding SSYT. The number of boxes in the SSYT with the number $l$ in row $k$ is $m_{k,l}-m_{k,l-1}$. The total number of boxes with the number $l$ is therefore the difference of the row sums $\sum_k (m_{k,l}-m_{k,l-1})$ (where $m_{k,l}=0$ if $k>l$). More systematically, in order to convert a GT pattern to an SSYT, we start from bottom-most row of the GT pattern and create a partial SSYT with one row of $m_{1,1}$ boxes labelled $1$. Next, we add $m_{1,2}-m_{1,1}$ boxes to this row labelled $2$ and put $m_{2,2}$ boxes in the second row labelled $2$. Continuing in this way, we add $m_{1,l}-m_{1,l-1}$ boxes to the first row labelled $l$, $m_{2,l}-m_{2,l-1}$ boxes in row two labelled $l$ etc. The in-betweenness conditions guarantee that the skew tableau with boxes labelled $l$ is a \emph{horizontal strip} i.e., a skew tableau with at most one box in each of its columns. Now, to convert from an SSYT to a GT pattern, we can use the fact that the number of boxes labelled $l$ in the $k^{th}$ row is $m_{k,l}-m_{k,l-1}$ and fill in the $k^{th}$ diagonal. 

The generators of the Lie algebra of \SU(d) are defined as follows.
\begin{align}
&J_0^{(l)} = \frac{1}{2}(E^{l,l}-E^{l+1,l+1})\,,\\
&J_{+}^{(l)}=E^{l,l+1}\,,\\
&J_{-}^{(l)}=E^{l+1,l}\,,
\end{align}
where $E^{k,l}$ is the matrix with a one in the $(k,l)^{th}$ position and zeros everywhere else and $1\leq l\leq d-1$. The matrix elements of these generators can be expressed in the GT basis. The action of these elements on a basis vector corresponding to a GT pattern $\ket{M}$ is given by
\begin{align}\label{Eq:GTbasis_action}
&J_0^{(l)}\ket{M} = \left[\sum_{k=1}^l m_{k,l} - \frac{1}{2}(\sum_{k=1}^{l+1} m_{k,l+1} + \sum_{k=1}^{l-1}m_{k,l-1})\right] \ket{M}\,,\\
&\bra{M+\delta_{k,l}}J_{+}^{(l)}\ket{M} = \left( -\frac{\Pi_{k'=1}^{l+1}(m_{k',l+1}-m_{k,l}+k-k')\Pi_{k'=1}^{l-1}(m_{k',l-1}-m_{k,l}+k-k'-1)}{\Pi_{k'=1,k'\neq k}^{l}(m_{k',l}-m_{k,l}+k-k')(m_{k',l}-m_{k,l}+k-k'-1)}\right)^{1/2}\,,\nonumber\\
&\bra{M-\delta_{k,l}}J_{-}^{(l)}\ket{M} = \left( -\frac{\Pi_{k'=1}^{l+1}(m_{k',l+1}-m_{k,l}+k-k'+1)\Pi_{k'=1}^{l-1}(m_{k',l-1}-m_{k,l}+k-k')}{\Pi_{k'=1,k'\neq k}^{l}(m_{k',l}-m_{k,l}+k-k'+1)(m_{k',l}-m_{k,l}+k-k')}\right)^{1/2}\,.\nonumber
\end{align}
Here $\delta_{k,l}$ is a triangle of numbers like a GT pattern with zeros everywhere and a one in the $k^{th}$ diagonal and $l^{th}$ row. It is not a valid GT pattern on its own. In the above formulae, only those $M\pm\delta_{k,l}$ are considered that are valid GT patterns.

\subsection{Schur-Weyl duality}
Schur-Weyl duality refers to the fact that the actions of the symmetric group and the unitary group on $V^{\otimes n}$ are full centralizers  of each other. A good description of Schur-Weyl duality can be found in the book by Goodman and Wallach \cite{goodman2000representations}. In terms of representations, this can be written as follows. Suppose we pick the representation of the symmetric group on $V^{\otimes n}$ and block diagonalize it to obtain irreps as follows
\begin{equation}
V^{\otimes n} \cong \bigoplus_\lambda W_\lambda\otimes V_\lambda\,,
\end{equation}
where $\lambda$ runs over all the irreps of the symmetric group, $V_\lambda$ is the irrep space of the irrep $\lambda$ and $W_\lambda$ is the multiplicity space on which the symmetric group acts trivially. Now consider the following action of the unitary group: $U^{\otimes n}$ i.e., the diagonal action where the same unitary acts on each copy of the vector space $V$. This action clearly commutes with the action of the symmetric group. This means that in terms of representations, $W_\lambda$ is a representation of the unitary group for each $\lambda$. Schur-Weyl duality essentially asserts that $W_\lambda$ is not just a representation, but rather an \emph{irreducible representation} of the unitary group. Stated in yet another way, the same unitary transformation that block diagonalizes the symmetric group representation into irreps also block diagonalizes the unitary group representation into irreps. In order to write this in terms of a basis, let us first pick a basis for $V$ as $\{\ket{1},\dots ,\ket{d}\}$. Then a basis for $V^{\otimes n}$ is the set $\{\ket{i_1,\dots ,i_n}\}$, where $i_k\in [d]$. The basis after block diagonalization can be written as $\ket{\lambda, i, j}$, where $\lambda$ labels the symmetric (or unitary) group irrep, $i$ is an index for a basis of the symmetric group irrep and $j$ indexes the unitary group irrep basis. In terms of this, the (strong) Schur transform can be defined as the unitary transformation that changes the basis from the computational one i.e., $\ket{i_1,\dots ,i_n}$ to the block diagonal one $\ket{\lambda, i, j}$.

The label $i$ of the unitary group irrep is essentially a GT pattern or equivalently a SSYT and the label $j$ is a SYT since it corresponds to the Young-Yamanouchi basis element of the symmetric group. The GT basis of the unitary group irrep $\lambda$ consists of a highest weight given by $(\lambda_1,\dots,\lambda_d)$ and the highest weight vector is represented by the SSYT of shape $\lambda$ and content is also $\lambda$ i.e., its first row contains all ones, second row is all twos etc. In fact, the weight of any basis vector is the content of the SSYT. Therefore, all the basis vectors corresponding to SSYTs of a fixed content $\mu$ are degenerate and the multiplicity of this weight space is $K_{\lambda,\mu}$. The correspondence between the Kostka number $K_{\lambda,\mu}$ and the multiplicity of the weight space with content $\mu$ can be shown combinatorially (see for instance the book by Stanley \cite{Stanley}).

\subsection{Permutation modules of the symmetric group}

Now, let us look at so called \emph{permutation modules} of the symmetric group, which are useful in understanding the structure of the $S_n$ representation on the space $V^{\otimes n}$. For a more careful treatment, see the book by Sagan \cite{Sag13}. First, define the \emph{type} of any $n$-tuple $E=(e_1,e_2,\dots ,e_n)$ with $e_i\in [d]$ to be an $n$-tuple $\mathcal{T}(E)=(t_1,\dots,t_n)$, where $t_i$ is the number of occurrences of $i$ in $E$. Clearly, $\sum_i t_i=n$ and so corresponding to any $E$ is a partition $\mu(E)$. Given a type $\mathcal{T}$, denote by $W(\mathcal{T})$, the set of all $n$-tuples of that type. This set can be obtained by starting with the tuple $E_0(\mathcal{T})=(1,\dots,1,2,\dots,2,\dots ,n,\dots,n)$, where there are $t_i$ elements labeled $i$ and then applying all possible permutations to it. Now, a permutation module corresponding to $\mathcal{T}$ is the representation of $S_n$ on the vector space with basis as the set $W(\mathcal{T})$. This basis comprises of vectors of the form $\ket{E}=\ket{e_1,e_2,\dots ,e_n}$, where $\mathcal{T}(E)=\mathcal{T}$. Let $\mu=\mu(\mathcal{T})$ denote the associated partition i.e., the tuple obtained by arranging the non-zero elements of $\mathcal{T}$ in decreasing order. 

It turns out that the representation of $S_n$ described above is an induced representation. It is induced from the trivial representation of a particular subgroup to the full group $S_n$. This subgroup is denoted $Y_\mathcal{T}$ and called a \emph{Young} subgroup. The Young subgroup is the stabilizer of $E_0$ i.e., all possible permutations of $S_n$ that preserve $E_0$. So the permutation module is $P(\mathcal{T})=\Ind{Y_\mathcal{T}}{S_n} \triv$. It turns out that the representation of $S_n$ on $V^{\otimes n}$ is just the direct sum of the permutation modules $P(T)$, where the sum is over all possible types $\mathcal{T}$. These permutation modules are reducible in general and decompose into irreps $\lambda$ of $S_n$. However, not all irreps appear in this decomposition. Only those $\lambda$ which dominate $\mu(\mathcal{T})$ (in a certain ordering defined below) appear in the decomposition. Their multiplicities are called Kostka numbers and are denoted $K_{\lambda\mu}$. The dominance order on the irreps or Young diagrams is the following. A Young diagram or a partition $\lambda$ is said to dominate $\mu$ if $\lambda_1+\dots +\lambda_k$ $\geq \mu_1+\dots +\mu_k$ for all $k\geq 1$.

Let us now look at the structure of the multiplicity space of any irrep in a permutation module. We would like to understand this space and its basis since, in the dual version of the Schur transform, this space leads to the irrep space of the unitary group. As mentioned earlier, the dimension of the multiplicity of $\lambda$ in the permutation module of the partition $\mu$ is $K_{\lambda\mu}$. This space has a basis in terms of semi-standard Young tableau (SSYT) of shape $\lambda$ and content $\mu$ (both of which are partitions of $n$) i.e., a Young diagram of shape $\lambda$ filled with $\mu_1$ ones, $\mu_2$ twos etc., such that the numbers are strictly increasing in the columns and weakly increasing in the rows.\remove{For example, if $\lambda=(2,2)$ and $\mu=(2,1,1)$, then ${\scriptsize\Yvcentermath1\young(11,23)}$ is a SSYT.} As a special case, when $\lambda=\mu$, we have $K_{\lambda,\lambda}=1$. In other words, there is only one SSYT with content and shape given by the same Young diagram. For $\lambda=(2,2)$, it is ${\Yboxdim{13pt}\Yvcentermath1\young(11,22)}$ and it turns out that such SSYTs lead to highest weight vectors in the Gelfand-Tsetlin basis of the unitary group.

\section{Quantum Fourier transforms}\label{Sec:QFTs}
\subsection{Precision of quantum transforms}\label{subsec:precision}
The precision of a unitary operator can be defined as follows. Given a target unitary $V$, $U$ is called an approximation to a precision $\epsilon$ if
\be
\sup_{\ket{\psi}\neq 0}\frac{||(U-V)\ket{\psi}||}{||\ket{\psi}||}\leq \epsilon,
\ee
where $||\ket{\psi}||$ is the norm of the state $\ket{\psi}$. It can be shown \cite{kitaevASP} that a computation consisting of a sequence of $L$ $\epsilon$-approximate unitaries followed by a measurement that has a error probability $\delta$ has an overall error probability $\leq \delta + 2L\epsilon$.

When one has a $m\times m$ unitary matrix whose entries can be efficiently computed, one can use the Solovay-Kitaev theorem to $\epsilon$-approximate it by a sequence of gates from a universal gate set using $O(m^2\log^c(m^2/\epsilon))$ elementary operations. For a constant sized $m$, this is efficient in $\log(1/\epsilon)$. As we will see below, the QFT over the symmetric group $S_n$ can be done in time O(poly($n,\log(1/\epsilon$)).

\subsection{QFT over the symmetric group}\label{subsec:QFTSn}
Although implicit in steps of Beals' algorithm \cite{Beals97}, the dependence on the precision $\epsilon$ is not written explicitly. Here we show that it is poly($\log(1/\epsilon)$). We briefly explain the steps in Beals' algorithm for a quantum Fourier transform over $\mathbb{C}(S_n)$ and the labeling of the Fourier basis used in it. The algorithm proceeds by reducing each element of the symmetric group into a product of adjacent transpositions. The set of adjacent transpositions $\{(12),(23),\dots ,(n-1n)\}$ generate the group and hence any element can be written as a product of adjacent transpositions. Beals' algorithm uses subgroup adapted bases, strong generating sets with small adapted diameter (techniques that have been generalized to several other groups in \cite{MRR06}). By inductively constructing the Fourier transform on the subgroup tower $\{S_n, S_{n-1},\dots, S_1\}$, the algorithm converts from the group basis $\ket{g}$ to the basis $\ket{\lambda, i,j}$. The indices $i$ and $j$ label the multiplicity space and irrep space, which are of the same dimension since this is the regular representation. Therefore, they are both labeled by SYTs defined above. 

They can also be labeled by paths in the \emph{Bratelli} diagram, which is a rooted tree with nodes at each level $n$ corresponding to all the inequivalent irreps of $S_n$. In this tree, there are edges between a node or irrep at level $n$ (say $\rho$) and a node or irrep at level $n-1$ (say $\sigma$) if $\sigma$ is contained in the restriction of $\rho$ to $S_{n-1}$. The multiplicity of this edge is equal to the multiplicity of $\sigma$ in $\rho$. We will use this algorithm as a subroutine below. We will denote a QFT over $S_n$ as QFT$(S_n)$ in the following and a QFT over any Young subgroup $Y_\mathcal{T}$ by QFT$(Y_\mathcal{T})$.

The main steps in the algorithm can be summarized as proceeding from $S_1$ to $S_n$ along the tower,
\begin{enumerate}
\item Embed an irrep of $S_k$ into an irrep of $S_{k+1}$.
\item Apply the unitary $\rho(t)$, where $\rho$ is an irrep of $S_{k+1}$ and $t$ is a transversal i.e., an element of $S_{k+1}/S_k$.
\item Sum over all cosets of $S_k$ in $S_{k+1}$.
\end{enumerate}
Each of these steps involves applying a unitary transform that is sparse (as shown in \cite{Beals97}) with only a constant number of non-zero entries that can be calculated efficiently. Using standard results described in the previous subsection, we can approximate the QFT in poly($\log(1/\epsilon)$).

\subsection{Fourier transform over induced representations}
As mentioned earlier, the usual Fourier transform is a unitary operator, which changes the basis from the basis of group elements $\{\ket{g} \mid g\in G\}$ to the block diagonal form given in~\eqref{Eq:regular-rep}. In this subsection, we use this transform to construct a Fourier transform for induced representations i.e., a transform that block diagonalizes induced representations. It turns out that the Fourier transform for induced representations allows us to construct the dual Schur transform.

Suppose we have an induced representation from $H$ to $G$ of an irreducible representation $\sigma$ of $H$ i.e., we have $\Ind{H}{G} \sigma$. The computational basis for this space can be written as $\ket{t,v}$, where $t$ is an element of the transversal and $v$ is a basis vector in the representation space of $\sigma$. This induced representation is in general reducible as a representation of $G$ and can be decomposed as a sum of irreducible representations. The multiplicity of each irreducible representation $\rho$ in the decomposition is equal to, using Frobenius reciprocity, the multiplicity of $\sigma$ in the restriction of $\rho$ to $H$. We now describe how one can perform this transform.

The change of basis we want to implement is from the basis labeled  $\ket{t,v}$ to a block diagonal basis labeled, say $\ket{\lambda,i,j}$. Here $\lambda$ labels the irreps that appear in the decomposition and $i$ and $j$ are the multiplicity and the irrep space basis vectors. 
\begin{enumerate}
\item First, we append an ancilla register of size $|H|/\text{dim}(\sigma)$ (which is an integer) to the initial state so that we have a register of size $H$ (excluding the transversal register). We now embed $\ket{v}$ into this register of size $H$ as $\ket{\sigma,u,v}$, where $u$ labels the multiplicity space of dimension dim($\sigma$). 
\item Next, perform the inverse QFT over $H$ to get the group basis $\ket{h}$. So including the transversal register, we have $\ket{t,h}$. Now, write this as $\ket{g}$. For the symmetric group, it turns out that this is easy since both $t$ and $h$ are specified in terms of adjacent transpositions. However, if the group basis for certain groups is defined in a complicated way, then this transform might be non-trivial.
\item Now perform the QFT over $G$ to obtain the basis $\ket{\lambda,i,j}$. The label $i$ is the multiplicity index, which runs over the entire dimension of $\lambda$. However, for an induced representation, it would run over a smaller set in general. Similarly for the irrep label. If the irrep and the multiplicity index labels can be ordered in such a way that the ones that occur in the decomposition are higher in the ordering, then we can easily return the ancilla register and obtain a clean transform. For the goal of block diagonalization, this ordering is not so important. However, we will see that this can be done for the case of the symmetric group. 
\end{enumerate}

\subsection{Quantum Fourier transform over permutation modules}
As explained above, permutation modules are induced representations from the trivial representation of some Young subgroup $Y_\mathcal{T}$ to $S_n$. The Young subgroup $Y_\mathcal{T}$ can be regarded as the stabilizer group of some $n$ tuple $E=(e_1,\dots, e_n)$ of type $T$ with entries $e_i\in [d]$. This representation space is spanned by all possible permutations of $E$. In order to construct the Fourier transform over this space, we first re-write these vectors in a way that reflects the structure of the induced representation. In other words, we choose a transversal and think of elements of the transversal $\ket{t}$ as basis vectors. Now, we can use the algorithm to construct Fourier transforms over induced representations. For clarity, we make the steps involved more explicit here.

\begin{itemize}
\item First, take an ancilla of size $|Y_\mathcal{T}|$ (more precisely, $m=\lceil\log|Y_\mathcal{T}|\rceil$ qubits) with all the qubits in the  state $\ket{0}$. Here $|Y_\mathcal{T}|$ stands for the number of elements in the subgroup $Y_\mathcal{T}$. The state $\ket{0}^{m}$ is taken to be the label of the trivial irrep of $Y_\mathcal{T}$.
\item Then perform the inverse Fourier transform over $Y_\mathcal{T}$ to obtain the equal superposition over all the elements of $Y_\mathcal{T}$. This can now be thought of as a subspace of $\mathbb{C}(S_n)$ spanned by equal superpositions over elements of cosets of $Y_\mathcal{T}$ in $S_n$. 
\item On this space, we can perform the quantum Fourier transform over $S_n$. This produces the basis $\ket{\lambda,i,j}$ with $i$ and $j$ labeled by SYTs.
\end{itemize}
The set of irreps $\lambda$ that appear in the decomposition are the ones that dominate the Young diagram corresponding to $Y_\mathcal{T}$. Similarly, the multiplicity space, although embedded inside a space of dimension $d_\lambda$, does not have support on all the vectors. It will be supported only on $K_{\lambda\mu}$ many vectors. This space is the multiplicity of the trivial representation of the subgroup $Y_\mathcal{T}$ when $\lambda$ is restricted to $Y_\mathcal{T}$.

While the algorithm above does an essential block diagonalization, we would ideally like to have the multiplicity index label basis vectors of the right subspace rather than have it label the basis of a larger space. In order to do this, we need to change the basis inside multiplicity spaces to correspond to the trivial space under the action of the subgroup $Y_\mathcal{T}$. This follows from Frobenius reciprocity as mentioned earlier. It turns out this is also important to get to the Gelfand-Tsetlin basis in the dual Schur transform that we construct in the next section. To perform this base change in the multiplicity spaces, we can use generalized phase estimation (GPE) \cite{H05}. GPE is a generalization of Kitaev's phase estimation technique \cite{KSV02}. In this technique, one can block diagonalize any representation $\rho(g)$ if one can perform $\rho(g)$ with $g$ as control. The main reason why we need the GPE is to organize the multiplicity space into a basis that consists of vectors that transform trivially under $Y_\mathcal{T}$. We can use this primitive to re-organize the multiplicity space. 

The action of $Y_\mathcal{T}$ and $S_n$ in the multiplicity space is the so called \emph{right regular representation} $R$. This acts on Young tableau according to the Young orthonormal representation since for the right action of $S_n$, the multiplicity space is an irrep. The result of performing GPE on the basis vector ${\Yboxdim{13pt}\Yvcentermath1\young(12,34)}$ (say) for $Y_\mathcal{T}=S_2\times S_2$ would be to stabilize it. In other words, performing $(12)$ or $(34)$ does not affect the vector. This would then correspond to the SSYT ${\Yboxdim{13pt}\Yvcentermath1\young(11,22)}$. Some vectors do not appear in the multiplicity space of the trivial irrep of $Y_\mathcal{T}$ in GPE. For example, the vector ${\Yboxdim{13pt}\Yvcentermath1\young(13,24)}$ would be taken to zero by the action of $Y_\mathcal{T}$. These statements are proved later in generality. In order to attain this change of basis, we need to perform GPE. At the end of GPE, instead of standard Young tableau labeling the basis vectors of the multiplicity space, we would have $\lambda(Y_\mathcal{T})|T\rangle$, where $T$ is a SYT. Here 
\be
\lambda(Y_\mathcal{T})=\frac{1}{|Y_\mathcal{T}|}\sum_{g\in Y_\mathcal{T}}\lambda(g)\,,
\ee
and $\lambda(g)$ is the representation corresponding to the partition $\lambda$ evaluated at the group element $g$.

In order to have a controlled application inside the multiplicity space, we can first apply a controlled right multiplication and then apply the quantum Fourier transform. We can just combine GPE with the algorithm described above to get a decomposition of the multiplicity space. However, for completeness, we explicitly list out all the steps below. The algorithm for GPE and its performance guarantee are as follows \cite{H05}.
\begin{algobox}{0.9}
\GPEtxt$_\rho (G)$\label{alg:GPE}

\emph{Inputs:} A quantum state $\ket{\psi}$ in the representation space of $\rho$ of a group $G$.

\emph{Blackbox:} The ability to perform controlled multiplication in the representation $\rho$.

\emph{Outputs:} The outcome $\lambda$ of an irrep of $G$ with probability $p_\lambda = \bra{\psi}\Pi_\lambda\ket{\psi}$.

\emph{Runtime:} $2T_{QFT(G)}+T_{C_\rho}$, where $T_{QFT(G)}$ is the time to perform a QFT over the group $G$ and $T_{C_\rho}$ is the time to perform controlled multiplication in the representation $\rho$.
\begin{enumerate}
  \item Take an ancilla register consisting of $\lceil\log|G|\rceil$ qubits and create an equal superposition over elements of $G$.
  \item Perform $C_\rho=\sum_g \ket{g}\bra{g}\otimes \rho(g)$.
  \item Perform a QFT($G$) to the ancilla register.
  \item Measure the irrep label of the ancilla register.
\end{enumerate}
\end{algobox}

The black box in the above algorithm can be made explicit in the following algorithm. The representation $\rho$ turns out to be the usual right regular representation and can be efficiently implemented. The overall algorithm to block diagonalize permutation modules is the following. This algorithm is potentially applicable to other problems where one needs to block diagonalize induced representations and could be of independent interest. In the following, we will use \GPE ~as a unitary module without the measurement step (step 4 above). The algorithm \QFTPermMod~described next is essentially the quantum part of the algorithm \DualSchur~described in the next section. To make things clear, we have included a block diagram of the algorithm \QFTPermMod~in Fig. \ref{QFTPMod}.

\begin{algobox}{0.9}
\QFTPermModtxt$(Y_\mathcal{T})$\label{alg:QFTPermMod}

\emph{Inputs:} A quantum register $A$ with the computational basis given by elements of the transversal of $Y_\mathcal{T}$ in $S_n$.

\emph{Outputs:} Quantum registers $\ket{\lambda,i,j}$ corresponding to the block diagonalization of the induced representation of the trivial irrep of $Y_\mathcal{T}$ to $S_n$.

\emph{Runtime:} O(poly$(n,\log\epsilon^{-1})$).

\begin{enumerate}
  \item Take an ancilla register $B$ consisting of $\lceil\log|Y_\mathcal{T}|\rceil$ qubits and create equal superposition over elements of $Y_\mathcal{T}$.
  \item Perform GPE$_R(Y_\mathcal{T})$ on register $AB$ where $R$ is the right regular representation of $S_n$.
  \item Perform a QFT($S_n$) on the registers $AB$.
\end{enumerate}
\end{algobox}

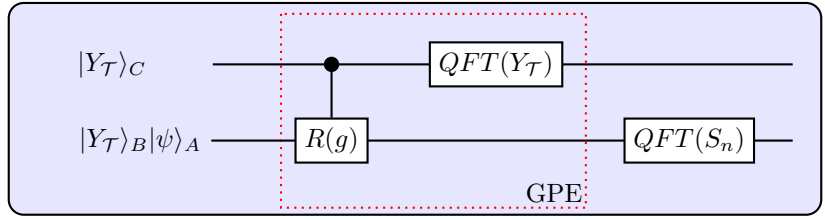
\begin{figure}
\centering
    \begin{tikzpicture}[thick]
    \tikzstyle{operator} = [draw,fill=white,minimum size=1.5em] 
    \tikzstyle{control} = [draw,fill,shape=circle,minimum size=5pt,inner sep=0pt]
    \tikzstyle{target} = [draw,shape=circle,minimum size=5pt,inner sep=0pt]
    \tikzstyle{neg_control} = [draw,shape=circle,fill=white,minimum size=5pt,inner sep=0pt]
    \tikzstyle{surround} = [fill=blue!10,thick,draw=black,rounded corners=2mm]
    \matrix[row sep=0.4cm, column sep=0.8cm] (circuit) {
    \node (q1) {$\ket{Y_\mathcal{T}}_C\hspace{0.75cm}$}; &[-0.5cm] &
    \node[control] (P12) {}; & 
    \node[operator] (P13) {$QFT(Y_\mathcal{T})$}; & 
    &[-0.3cm]
    \coordinate (end1); \\
     
    \node (q2) {$\ket{Y_\mathcal{T}}_B\ket{\psi}_A$};  &&
    \node[operator] (P23) {$R(g)$}; & &
    \node[operator] (P22) {$QFT(S_n)$}; &
    \coordinate (end2);\\
    };
    \node [draw=none,fill=none,below] at ($(P22.south)+(0.3,-0.3)$) (new) {};
    \node [draw=none,fill=none,below] at ($(circuit.north east)-(0cm,1.4cm)$) (right) {};
    \node [draw=none,fill=none,below] at ($(q1)+(-1.5,0.7)$) (q1p) {};
    \begin{pgfonlayer}{background}
        \node[surround] (background) [fit = (q1p) (new) (right)] {};
        \draw[red,thick,dotted] ($(P12.north west)+(-0.6,0.6)$)  rectangle ($(P13.south east)+(0.3,-1.6)$);
        \node at ($(P13.south east)+(-0.12,-1.4)$) {GPE};
        \draw[thick] (q1) -- (end1)  (q2) -- (end2) (P12) -- (P23);
    \end{pgfonlayer}
    \end{tikzpicture}\caption{Circuit diagram of \QFTPermMod. Here $\ket{Y_\mathcal{T}}$ stands for an equal superposition over the elements of the subgroup $Y_\mathcal{T}$, the registers $B$ and $C$ contain the state $\ket{Y_\mathcal{T}}$, register $A$ contains an arbitrary state $\ket{\psi}$ with support on the transversal. The dotted red box shows the part that is GPE.}
    \label{QFTPMod}
    \end{figure}

\begin{theorem}\label{theorem:QFTPermMod}
The above algorithm \QFTPermMod ~performs a block diagonalization of the permutation module in time O(poly$(n,\log\epsilon^{-1})$).
\end{theorem}
\begin{proof}
The proof can be broken into three parts.
\begin{enumerate} 
\item We show that for states of the form $\sum_t a_t \ket{t}$, performing GPE and measuring the irrep of $Y_\mathcal{T}$ would always give the trivial irrep. 
\item We show that when the trivial irrep of $Y_\mathcal{T}$ appears in the measurement, the computational basis of the multiplicity space is rotated from SYTs to (the normalized version of) $\lambda(Y_\mathcal{T})|T\rangle$, where $T$ is a SYT. 
\item Then we show that the states $\lambda(Y_\mathcal{T})|T\rangle$ are in one-to-one correspondence to SSYTs whenever the states are non-zero. 
\end{enumerate}

\emph{Part (i)}

To show that we always get the trivial irrep of $Y_\mathcal{T}$ for states of the form $\sum_t a_t \ket{t}$, where $t$ is the transversal of $Y_\mathcal{T}$ in $S_n$, we track the state through the steps of the algorithm.
\begin{enumerate}[leftmargin=*]
\item Suppose we had the state $\sum_t a_t\ket{t}^A$, where $t$ is an element of the transversal of $Y_\mathcal{T}$ in $S_n$. We take an ancilla register $B$ consisting of $\log|Y_\mathcal{T}|$ qubits to get $\sum_t a_t\ket{t}^A\ket{0}^B$.
  
\item Perform inverse QFT over $Y_\mathcal{T}$ on the $B$ register to obtain an equal superposition over group elements of $Y_\mathcal{T}$. This allows us to view the registers $A$ and $B$ together as the group basis of $S_n$. This step takes O(poly$\log|Y_\mathcal{T}|$) time since the QFT over $Y_\mathcal{T}$ can be done efficiently for Young subgroups (as they are direct products of symmetric groups). We now have the state $\sum_{t,h} b_t\ket{t}^A\ket{h}^B$. Here $h$ runs over all the elements of $Y_\mathcal{T}$ and $b_t=a_t/\sqrt{|Y_\mathcal{T}|}$.

\item Perform GPE, which consists of the following steps. This takes O(poly$\log|Y_\mathcal{T}| +$ poly($n$)) time since controlled $R$ operations can be done in poly $n$ time and there are two QFTs over $Y_\mathcal{T}$.
\begin{itemize}[leftmargin=*]
\item Take a register $C$ of $\log|Y_\mathcal{T}|$ qubits initialized to $\ket{0}$ and obtain an equal superposition over $Y_\mathcal{T}$ in it. This gives us the state 
\begin{equation}
\sum_{t,h_1,h_2} c_t\ket{t}^A\ket{h_1}^B\ket{h_2}^C\,,
\end{equation}
where $h_1$ and $h_2$ run over all elements of $Y_\mathcal{T}$ and $c_t=b_t/\sqrt{|Y_\mathcal{T}|}$.
\item Conditioned on the register $C$, perform a controlled right multiplication on the group basis of $S_n$ in registers $AB$ i.e., perform 
\begin{equation}
\sum_{h\in Y_\mathcal{T}}\ket{h}\bra{h}^C\otimes R(h)^{AB}\,,
\end{equation}
where $R$ is the right multiplication in $S_n$. This gives the state 
\be
\sum_{t,h_1,h_2} c_t\,R(h_2)^{AB}\,(\ket{t}^A\ket{h_1}^B)\,\ket{h_2}^C\,.
\ee
This state can be rewritten as 
\be
\sum_{t,h_1,h_3} c_t\,(\ket{t}^A\ket{h_3}^B)\,\ket{h_1^{-1}h_3}^C\,,
\ee
where we have replaced $h_1h_2$ as $h_3$.

\item Perform QFT($Y_\mathcal{T}$) on $C$. This gives us the following state 
\begin{equation}
\sum_{\mu,k,l,h_1,h_3,t} \frac{\sqrt{d_\mu}}{|Y_\mathcal{T}|}\,[\mu(h_1^{-1}h_3)]_{k,l}\, c_t\,\ket{t,h_3}^{AB}\ket{\mu,k,l}^C\,,
\end{equation}
where $\mu$ runs over all the irreps of $Y_\mathcal{T}$ and $k$ and $l$ run over its dimension. The sum over $h_1$ forces $\mu$ to be the trivial irrep of $Y_\mathcal{T}$. Thus, for states of the form $\sum_t a_t \ket{t}$, we will always get the trivial irrep when we measure $\mu$.

\end{itemize}

\end{enumerate}

\emph{Part (ii)}\\
If we had done a QFT over $S_n$ and not performed GPE before that, we would have had the basis of the multiplicity space labeled by SYTs. Here we show that, after performing GPE and when $\mu$ (the irrep label of $Y_\mathcal{T}$) is trivial, the basis of SYTs is rotated to $\lambda(Y_\mathcal{T})\ket{T}$, where $T$ is a SYT of shape $\lambda$. Starting with a state of the form $\ket{th_1}$ and performing GPE would give us the state
\be
\sum_{\mu,k,l,h_2} \frac{\sqrt{d_\mu}}{|Y_\mathcal{T}|}\,[\mu(h_2)]_{k,l}\,R(h_2)\ket{t,h_1}^{AB}\ket{\mu,k,l}^C\,,
\ee
Next we perform a QFT over $S_n$ on this state to get the following state
\be\label{Eq:GT_basis}
\sum_{h_2,\lambda,\mu,T_1,T_2,k,l}\frac{\sqrt{d_\mu d_\lambda}}{|Y_\mathcal{T}|n!}[\lambda(th_1)]_{T_1,T_2}\,[\mu(h_2)]_{k,l}(\ket{\lambda}\otimes \lambda(h_2)\ket{T_1}\otimes\ket{T_2})^{AB}\ket{\mu,k,l}^C \,,
\ee
where $T_1$ and $T_2$ are SYTs. Since the right regular representation acts only on the first register, we have the $\lambda(h_2)$ acting only on $\ket{T_1}$. When $\mu$ is trivial the basis state $\ket{T_1}$, which corresponds to a SYT gets taken to $\lambda(Y_\mathcal{T})\ket{T_1}$.

\emph{Part (iii)}\\
We show next that $\lambda(Y_\mathcal{T})\ket{T}$ can be identified with semistandard Young tableau of shape $\lambda$ and content defined by $Y_\mathcal{T}$. In other words, if $Y_\mathcal{T}=S_{X_1}\times S_{X_2}\times\dots \times S_{X_k}$ for some $k$, then the content is $|X_1|$ $1$s, $|X_2|$ $2$s etc. Also note that the sets $X_i$ consist of consecutive integers. The SSYT associated with $\lambda(Y_\mathcal{T})\ket{T}$ is the one where all the integers in $X_i$ are replaced by $i$. This is a valid SSYT if no two integers in $X_i$ are in the same column or equivalently, every column has at most one element of $X_i$. If we isolate the boxes in the SYT labeled by elements of $X_i$ and they have the property above, such a skew Young diagram is called a \emph{horizontal strip} as described in section \ref{subsec:subgroup_adapted_bases}. 

We now only need to show that for every $X_i$ that if the boxes numbered with elements of $X_i$ form a horizontal strip, then $\lambda(Y_\mathcal{T})\ket{T}$ is non-zero and it is zero otherwise. To show this, we focus on a single $X_i$ and show that if there are two elements of $X_i$ in the same column, then $\lambda(X_i)\ket{T}=0$. This is done next in lemma~\ref{lemma:horizontal_strip}. Finally, the claim of the dependence on $\epsilon$ follows from the fact the each of the steps in the algorithm (including the group multiplications and quantum Fourier transforms) can be done with O(poly$\log\epsilon^{-1}$) elementary gates based on the results described in section~\ref{subsec:precision} and~\ref{subsec:QFTSn}.

\end{proof}
\begin{lemma}\label{lemma:horizontal_strip}
Let $A$ be a set of consecutive integers $\{a_1,\dots a_k\}$, $S_A$ be the symmetric group of size $|A|!$ permuting the elements of $A$. Let $\ket{T}$ be a SYT of shape $\lambda$ with entries that include the set $A$. Let $\lambda(S_A)=\sum_{\pi}\lambda(\pi)$, where $\pi$ runs over elements of $S_A$. Assume that $T$ contains two elements of $A$ in the same column. Then, $\lambda(S_A)\ket{T}=0$, where $\lambda(\pi)$ is the Young orthogonal representation on SYTs.
\end{lemma}
\begin{proof}
Since the elements of $A$ are consecutive integers, we can assume without loss of generality that there are two elements of $A$ that appear in the same column and in consecutive rows. Let the elements be $i$ and $j$ with $j>i$. We first show that if $j=i+1$, then $\lambda(S_A)\ket{T}=0$ and then reduce the general case to this one.

So assume now that $i$ and $i+1$ are in the same column (they have to be in consecutive rows). The action of the transposition $(i,i+1)$ is $\lambda((i,i+1))\ket{T}=-\ket{T}$. Therefore, $\lambda(e+(i,i+1))\ket{T}=0$, where $e$ is the identity element. It is easy to see that for any set $K=\{i_1,i_2\dots i_r\}$, the symmetric group algebra element that is a sum of all possible permutations of the elements of $K$ can be written as follows.
\be
S_K=((i_r,i_1)+\dots+(i_r,i_{r-1}))\dots ((i_3,i_1)+(i_3,i_2)+e)\,((i_2,i_1)+e)
\ee
While this factorization is dependent on the ordering of the elements, the overall group algebra element $S_K$ is independent of it. Using this and writing $A=\{i,i+1,i+2,\dots , a_k,a_1,a_2,\dots,i-1\}$, we obtain
\be
S_A=S_A^\prime \, ((i,i+1)+e)\,,
\ee
where in $S_A^\prime$, we collect the rest of the terms i.e., $S_A^\prime$ is a product of sums of transpositions coming from the factorization above. It is now easy to see that $S_A\ket{T}=0$ if $T$ contains $i$ and $i+1$ in the same column. For the general case, assume inductively that when $i$ and $j-1$ are in the same column, then $S_A\ket{T}=0$. Now, suppose that $i$ and $j$ are in the same column and in consecutive rows.


Consider the element $k$, where $k$ is the largest element between $i$ and $j$ such that $k$ is in a different row from $j$ and all elements between $k$ and $j$ are in the same row as $j$. For example, if $j-1$ is in a different row from $j$, then $k=j-1$. This, in particular, means that $k$ and $k+1$ are in different rows. If they are in the same column, then we are done. So assume that they are not in the same column. Then we have
Let $(k,k+1)$ be a transposition and let $\ket{(k,k+1)T}$ be the SYT with $k$ and $k+1$ interchanged. Since they are not in the same row or column, $(k,k+1)T$ is also a SYT. We have
\begin{align}
&(k,k+1)\ket{T}=a^T_k\ket{T}+b^T_k\ket{(k,k+1)T}\,,\\
&(k,k+1)\ket{(k,k+1)T}=b^T_k\ket{T}-a^T_k\ket{(k,k+1)T}\,,
\end{align}
where $a^T_k$ is the inverse of the Manhattan distance in $T$ between $k$ and $k+1$ and $b^T_k=\sqrt{1-(a^T_k)^2}$. Using these equations, we have
\be
\ket{T}=(A e + B(k,k+1))\ket{(k,k+1)T}\,,
\ee
where $A_k=\frac{1}{b^T_k}$ and $B_k=\frac{a^T_k}{b^T_k}$. Therefore, we have
\be
S_A\ket{T}=S_A(A_ke+B_k(k,k+1))\ket{(k,k+1)T}=(A_ke+B_k(k,k+1))S_A\ket{(k,k+1)T}\,,
\ee
where the last equality follows from the fact that both $e$ and $(k,k+1)$ commute with $S_A$. Now note that in $(k,k+1)T$, $k+1$ and $k+2$ are not in the same row or column. Continuing this process we get
\begin{align}
S_A\ket{T}=&(A_ke+B_k(k,k+1))(A_{k+1}e+B_{k+1}(k+1,k+2))\dots(A_{j-1}e+B_{j-1}(j-1,j))\nonumber\\
&S_A\ket{(j-1,j)\dots (k,k+1)T}\,.
\end{align}
It is easy to see that $\ket{(j-1,j)\dots (k,k+1)T}$ is a SYT where $i$ and $j-1$ are in the same column and consecutive rows. By the induction hypothesis, we have $S_A\ket{(j-1,j)\dots (k,k+1)T}=0$.

\end{proof}

\section{Dual algorithm for the Schur transform}\label{Sec:Dual_Schur_transform}
We are now ready to describe our dual algorithm for the Schur transform using the above tools. It involves essentially two main steps. The first is a block diagonalization into permutation modules and the second is a block diagonalization of each permutation module into irreps using the above transform. The algorithm is as follows.
\begin{algobox}{0.9}
\DualSchurtxt$(n,d,\epsilon)$\label{alg:DualSchur}

\emph{Inputs:} A quantum register $A$ with the computational basis given by $n$-tuples $\ket{e_1,\dots , e_n}$, where $e_j\in [d]$.

\emph{Outputs:} Quantum registers $\ket{\lambda,i,j}$, where $\lambda$ is the irrep label of the symmetric or unitary groups, $i$ is the irrep label of the symmetric group in the Young orthonormal  basis and $j$ is the irrep register of the unitary group in the Gelfand-Tsetlin basis.

\emph{Runtime:} O(poly$(n, \log d, \log \epsilon^{-1})$).

\begin{enumerate}
\item Map the entries of any basis vector that are greater than $n$ to entries inside $n$ and keep track of this mapping. For example, $\ket{e}=\ket{e_1,\dots,e_n}$ is now $\ket{p_e,\tilde{e}}$, where $\tilde{e}$ is a vector with entries in $[n]$ and $p_e$ is the map.
\item Convert $\ket{\tilde e}$ to $\ket{\mathcal{T},t}$, where $\mathcal{T}$ is the type and $t$ is the transversal element of the subgroup $Y_\mathcal{T}$ in $S_n$.
\item Conditioned on $\mathcal{T}$, apply \QFTPermMod ~to $\ket{t}$ and obtain the basis $\ket{\lambda, i,j}$, where $j$ is an SSYT with entries in $[d]$ and $i$ is an SYT.
\item The basis at this point is $\ket{\lambda, i,(\mathcal{T},j)}$. Use RSK to convert the pair $(\mathcal{T}, j)$ to a SSYT as described in section \ref{subsec:RSK}.
\end{enumerate}
\end{algobox}

\begin{theorem}
Given an $n$ fold tensor product of $d$ dimensional Hilbert spaces and accuracy $\epsilon$, the quantum algorithm \DualSchur ~runs in time O(poly $\log d, n, \log 1/\epsilon$) and performs the strong Schur transform i.e., performs the change of basis from the computational basis to the block diagonal basis $\ket{\lambda, i, j}$, where $\lambda$ is the symmetric or unitary group irrep, $i$ labels the Young-Yamanouchi basis vector in the symmetric group and $j$ labels the Gelfand-Tsetlin basis vector in the unitary group irrep.
\end{theorem}

\begin{proof}
To prove the runtime claim, we describe the steps in more detail with an example and bound the run time.
\begin{enumerate}
\item Given a basis vector in the computational basis, first we map the entries of the vector that are greater than $n$ to entries inside $n$ and keep track of the map. So a basis vector $\ket{e}=\ket{e_1,\dots,e_n}$ becomes $\ket{p_e,\tilde{e}}$, where $\tilde{e}$ is a vector with entries in $[n]$ and $p_e$ is the map. This map need not be global and can be specific to the vector $\ket{e}$. This can be done in poly$(n)$ steps. As a running example, let us consider $n=5$, $d=10$ and take the vector $\ket{5, 5, 10, 3, 9}$. In the first step, this gets mapped to $\ket{10\rightarrow 1, 9\rightarrow 2}\otimes\ket{5,5,1,3,2}$. This can be done in O(poly($n,\log d$)) steps.
\item Compute the type of the vector $\tilde{e}$ and the symmetric group element (as a product of a set of transpositions) needed to convert the standard basis vector i.e., where the entries appear in ascending order to $\tilde{e}$. The basis vector $\ket{e}$ is converted to $\ket{p_e,\mathcal{T},t}$, where $\mathcal{T}$ is the type and $t$ is the transversal element of the subgroup $Y_\mathcal{T}$ in $S_n$. For the example, we would have the standard basis vector as $\ket{1,1,2,3,4}$, the type is $(1,1,1,0,2)$ i.e., one $1$, one $2$, one $3$, zero $4$ and two $5$s. The transversal as a product of transpositions would be $(15)(52)(53)$. This takes poly$(n)$ time.
\item Recall that the register with $\ket{t}$ can be viewed as the induced representation of $Y_\mathcal{T}$ in $S_n$ i.e., a permutation module. Conditioned on the type $\mathcal{T}$, we can apply the Fourier transform for permutation modules to obtain the basis $\ket{p_e,\mathcal{T},\lambda, i,j}$. This step takes poly$(n,\log\epsilon^{-1})$ time.
\item Rewriting this as $\ket{\lambda, (p_e,\mathcal{T},i),j}$, we can convert $p_e,\mathcal{T},i$ into a SSYT by using the information in $p_e$ and $\mathcal{T}$ and rewriting the basis $i$, which consists of SSYT of shape $\lambda$ and content $\mathcal{T}$, to match the GT basis (we prove below that this is gives the GT basis).
\end{enumerate}
The claim of the dependence on $\epsilon$ follows from the fact the each of the steps in the algorithm can be done with O(poly$\log\epsilon^{-1}$) elementary gates  based on the results described in section~\ref{subsec:precision} and~\ref{subsec:QFTSn}. This shows that the overall runtime is O(poly($\log d, n, \log\epsilon^{-1}$)). 

To prove the theorem, we now need to prove that the basis obtained in the last step (labeled by SSYTs) is exactly the GT basis. The (unnormalized) basis states can be written in the form (from Eq. \ref{Eq:GT_basis})
\be
\ket{\lambda,p_e,\mathcal{T},i,j} = \sqrt{\frac{d_\lambda}{|G|/|Y_\mathcal{T}|}}\sum_g [\lambda(gY_\mathcal{T})]_{i,j}\ket{gY_\mathcal{T}}\,,
\ee
where $G$ is the symmetric group $S_n$, $Y_\mathcal{T}$ is a Young subgroup of the form $(S_{\mu_1}\times S_{\mu_2}\times \dots S_{\mu_k})$ for some $k\leq d$. This follows from the next lemma.
\end{proof} 

\begin{lemma}
Let $G$ be the symmetric group $S_n$ and $Y_\mathcal{T}$ be a Young subgroup corresponding to the type $\mathcal{T}$. Consider the unnormalized basis states written as 
\be
\ket{\lambda,p_e,\mathcal{T},i,j} = \sqrt{\frac{d_\lambda}{|G|/|Y_\mathcal{T}|}}\sum_g [\lambda(gY_\mathcal{T})]_{i,j}\ket{gY_\mathcal{T}}\,,
\ee
where the sum over $g$ is over the cosets of $Y_\mathcal{T}$ in $G$. When $\mathcal{T}$ runs over all possible types and $p_e$ runs over all possible subsets of $n$ letters from the alphabet $\{1,\dots, d\}$, the above set of vectors form the Gelfand-Tsetlin basis.
\end{lemma}

\begin{proof}
In order to prove this, we will verify that this basis is the subgroup adapted basis of the tower of subgroups $U(d)\supset U(d-1)\supset\dots \supset U(1)$. Equivalently, we will show that this basis makes the following decomposition block diagonal.
\be
V_\lambda\downarrow_{U(d-1)}^{U(d)}\simeq \bigoplus_{\mu} V_\mu\,, 
\ee
where $V_\lambda$ is the irrep of the unitary group $U(d)$ that corresponds to the shape $\lambda$ and $V_\mu$ is an irrep of $U(d-1)$ that corresponds to the shape $\mu$. The direct sum runs over all the shapes $\mu$ that differ from $\lambda$ by a horizontal strip. The shapes $\mu$ are all obtained from $\lambda$ by considering the various SSYTs that constitute the basis of the irrep space of the unitary group and removing the horizontal strips labelled $d$. This corresponds to all unitaries in $U(d)$ that fix the label $d$.

The (unnormalized) basis states can be written in the form 
\be
\ket{\lambda,i,j} = \sqrt{\frac{d_\lambda}{|G|/|H|}}\sum_g [\lambda(gH)]_{i,j}\ket{gH}\,,
\ee
where $G$ is the symmetric group $S_n$, $H$ is a Young subgroup of the form $(S_{\mu_1}\times S_{\mu_2}\times \dots S_{\mu_k})$ for some $k\leq d$. To avoid clutter, we replace $Y_\mathcal{T}$ above by $H$. We will call the tuple $\mu$, which can be permuted to a valid Young diagram, the content corresponding to $H$ since it is the content of the SSYTs that we obtain. The sum over $g$ runs over elements of the transversal of $H$ in $G$. With the embedding of the permutation module into the group algebra of $S_n$ done in the previous section, we can interpret $\ket{H}$ as the computational basis state $\ket{11\dots 1\, 2\dots 2\dots k\dots k}$ where there are $\mu_1$ $1$s, $\mu_2$ $2$s etc. The state $\ket{gH}$ as the basis state which is permuted according to the transversal element $g$. In the proof of theorem~\ref{theorem:QFTPermMod}, it was also shown the action of $H$ on the SYT $j$ leads to a SSYT. In order to show the above block diagonalization, we need to show that the action of $U(d-1)$ leaves the horizontal strip labeled $d$ intact. The action of the Lie algebra of $U(d-1)$ is generated by the operators $J_0^{(l)}, J_+^{(l)}$ and $J_-^{(l)}$ for $l=1,2,\dots,d-2$. These operators were defined in section \ref{subsec:GT}. We show that these operators acting on any $\ket{\lambda,i,j}$ leave the horizontal strip labeled $d$ intact. In other words, they give superpositions of states of the form $\ket{\lambda,i,j^\prime}$, where $j^\prime$ is a SSYT of shape $\lambda$ and whose horizontal strip labeled $d$ is the same as the one in $j$.

To show this, let us look at the action of $J_+^{(l)}$ for some $l$ in $[d-2]$. The action of $J_-^{(l)}$ is the Hermitian conjugate and so it suffices to look at $J_+^{(l)}$. We work with the unnormalized states given above.
\be
J_+^{(l)}\ket{\lambda,i,j} = J_+^{(l)}\sqrt{\frac{d_\lambda}{|G|/|H|}}\sum_g [\lambda(gH)]_{i,j}\ket{gH} = \sqrt{\frac{d_\lambda}{|G|/|H|}}\sum_{g^\prime} c_{g^\prime} \ket{g^\prime H^\prime}\,,
\ee
where $g^\prime$ runs over the transversal of $H^\prime$ in $G$ and $H^\prime$ is the Young subgroup that differs from $H$ in the number of $l$s and $l+1$s with one more $l$ and one less $l+1$ than $H$ i.e., if the content corresponding to $H$ is $\mu=(\mu_1,\dots ,\mu_k)$, then that of $H^\prime$ is $\mu^\prime=(\mu_1,\dots,\mu_l+1,\mu_{l+1}-1,\dots,\mu_k)$. The coefficient $c_{g^\prime}$ can be calculated to be
\be
c_{g^\prime}=[\lambda(g^\prime J H)]_{i,j}\,,
\ee
where the operator $J$ is defined as 
\be
J=\sum_{m=\sigma_l+1}^{\sigma_{l-1}+1} (\sigma_l+1, m)\,,
\ee
where $\sigma_l=\mu_1 +\dots +\mu_l$ and similarly $\sigma_{l+1}$. This shows that the new SSYT is a sum SSYTs obtained by taking a SYT $j$ and applying $H$ and an element of $J$. This means that from the SSYT $\tilde{j}$, we get a sum that involves SSYTs obtained by replacing some box labeled $l+1$ by $l$ as long as it gives a valid SSYT (since if two boxes labeled $l$ are in the same column, then by lemma~\ref{lemma:horizontal_strip}, that SSYT is taken to zero).

Now we show that the action of $J_0^{(l)}$ is as described in \ref{subsec:GT}. To see this note that $E^{l,l}$ acting on the computational basis counts the number of $l$s in the basis vector, which means that it counts the number of boxes labeled $l$ in the SSYT basis. Therefore since $J_0^{(l)}=(1/2)(E^{l,l}-E^{l+1,l+1})$, the action of $J_0^{(l)}$ on our basis is $(\mu_l-\mu_{l+1})/2$. Using the translation from SSYT to GT patterns described in section \ref{subsec:GT}, for a GT pattern $M$ defined as 
\begin{equation}
M=\begin{pmatrix}
m_{1,d} & & m_{2,d} & & \dots & & m_{d,d}\\
 & m_{1,d-1}& &  \dots & & m_{d-1,d-1} &\\
&& \ddots &&\reflectbox{$\ddots$}\\
&&m_{1,2}&m_{2,2}\\
&&\qquad \qquad \quad m_{1,1}
\end{pmatrix}\,,
\end{equation}
we have 
\be
\mu_l=\sum_{k=1}^{l} (m_{k,l}-m_{k,l-1})\,.
\ee
Thus the action of $J_0^{(l)}$ can be seen to be
\be
J_0^{(l)}\ket{M} = \left[\sum_{k=1}^l m_{k,l} - \frac{1}{2}(\sum_{k=1}^{l+1} m_{k,l+1} + \sum_{k=1}^{l-1}m_{k,l-1})\right] \ket{M}\,.
\ee
\end{proof}

\section{Conclusions}\label{Sec:Conclusions}
We have presented an efficient algorithm for a high dimensional Schur transform that runs in time O(poly($n$, $\log d$, $\log 1/\epsilon$)). This improves exponentially in the dimension over the prior work of Bacon, Chuang and Harrow \cite{BCH07}. As mentioned above, Harrow's thesis \cite{H05} contains a way to make the unitary group approach of \cite{BCH07} polynomial in $\log d$. Our algorithm is novel in that it uses the representation theory of the symmetric group rather than that of the unitary group. Another interesting feature is that it uses only the quantum Fourier transform and generalized phase estimation (which is also based on the QFT) and essentially no new tools. A potentially useful feature of this algorithm that could be a primitive for other problems is the circuit for a Fourier transform over induced representations. Several permutation modules, which are induced representations encode important problems that include element distinctness and collision finding. The subroutines to block diagonalize permutation modules could provide Fourier analytic algorithms to these problems and generalize to solve other problems which have permutational symmetry.

\section{Acknowledgments}
We would like to thank Aram Harrow for useful discussions and for explaining the algorithms in his thesis. We would like to thank John Wright for useful discussions and for comments on an earlier draft. We would also like to thank the anonymous referees of TQC for useful feedback. We are also grateful to the referees of Quantum for several suggestions that helped improve the paper.

\bibliography{Schur_refs}
\end{document}